\documentclass[12pt]{amsart}
\usepackage{amsmath, amsthm, latexsym, amssymb, amsfonts, epsfig, epsf, xcolor, hyperref}
\usepackage{mathtools}
\usepackage{bm,xcolor,tikz,hyperref}
\usepackage{bm}
\usepackage{blkarray}
\usepackage{booktabs}
\usepackage{framed}
\usepackage{multirow}

\numberwithin{equation}{section}
\newtheorem{theorem}{Theorem}[section]

\newtheorem{proposition}[theorem]{Proposition}
\newtheorem{corollary}[theorem]{Corollary}

\newtheorem{example}[theorem]{Example}
\newtheorem{remark}[theorem]{Remark}

\newcommand{\ev}{\textrm{ev}}
\newcommand{\mcc}{\mathcal{C}(\mathcal{S}, \mathcal{A})}
\newcommand{\polyset}{\mathcal{L}(\mathcal{A})}

\newcommand{\tr}{\textrm{Tr}}
\newcommand{\rs}{\textrm{RS}}
\newcommand{\grm}{\textrm{RM}}
\newcommand{\car}{\textrm{Car}}
\newcommand{\arm}{\textrm{ARM}}
\newcommand{\agrm}{\textrm{ARM1}}
\newcommand{\agrmm}{\textrm{ARM2}}
\newcommand{\mcart}{\textrm{M-CC}}
\newcommand{\dmcc}{\textrm{DM-CC}}
\newcommand{\acart}{\textrm{ACar}}
\newcommand{\acar}{\textrm{ACar1}}
\newcommand{\acarr}{\textrm{ACar2}}

\newcommand{\zm}{\mathbb{Z}_{\geq 0}^m}

\newcommand{\rmv}[1]{}
\def\F{\Bbb F}

\makeatletter
\def\blfootnote{\gdef\@thefnmark{}\@footnotetext}
\makeatother

\begin{document}
\title[Repairing decreasing and augmented codes]{Erasures repair for decreasing monomial-Cartesian and augmented Reed-Muller codes of high rate}

\author{Hiram H. L\'opez}
\address[Hiram H. L\'opez]{Department of Mathematics and Statistics\\ Cleveland State University\\ Cleveland, OH USA}
\email{h.lopezvaldez@csuohio.edu}
\author{Gretchen L. Matthews}
\address[Gretchen L. Matthews]{Department of Mathematics\\ Virginia Tech\\ Blacksburg, VA USA}
\email{gmatthews@vt.edu}
\thanks{The work of Hiram H. L\'opez was supported in part by an AMS--Simons Travel Grant. The work of Gretchen L. Matthews was supported in part by NSF under Grant DMS-1855136 and in part by the Commonwealth Cyber Initiative. (Corresponding author: Hiram H. L\'opez.)}
\author{Daniel Valvo}
\address[Daniel Valvo]{Department of Mathematics\\ Virginia Tech\\ Blacksburg, VA USA}
\email{vdaniel1@vt.edu}
\thanks{The portion of this work on augmented Reed-Muller codes was presented at ISIT 2021.}

\begin{abstract}
In this work, we present linear exact repair schemes for one or two erasures in decreasing monomial-Cartesian codes $\dmcc$, a family of codes which provides a framework for polar codes. In the case of two erasures, the positions of the erasures should satisfy a certain restriction. We present families of augmented Reed-Muller ($\arm$) and augmented Cartesian codes ($\acart$) which are families of evaluation codes obtained by strategically adding vectors to Reed-Muller and Cartesian codes, respectively. We develop repair schemes for one or two erasures for these families of augmented codes. Unlike the repair scheme for two erasures of $\dmcc,$ the repair scheme for two erasures for the augmented codes has no restrictions on the positions of the erasures. When the dimension and base field are fixed, we give examples where $\arm$ and $\acart$ codes provide a lower bandwidth (resp., bitwidth) in comparison with Reed-Solomon (resp., Hermitian) codes. When the length and base field are fixed, we give examples where $\acart$ codes provide a lower bandwidth in comparison with $\arm$. Finally, we analyze the asymptotic behavior when the augmented codes achieve the maximum rate.
\end{abstract}

\maketitle

\section{Introduction}
The design of linear exact repair schemes for evaluation codes began with the foundational work of Guruswami and Wootters in which they developed a repair scheme (GW-scheme) to efficiently repair an erasure in a Reed-Solomon ($\rs$) code ~\cite{GW}. 
This work served as motivation for linear exact repair schemes for algebraic geometry codes \cite{JLX} and Reed-Muller codes \cite{RM}. In each of these instances, codes are considered over an extension field whose elements may be represented using subsymbols, meaning elements of a smaller base field. Erasure recovery is accomplished using subsymbols rather than the symbols themselves. Under certain conditions, these new schemes require less information than standard approaches to repair. In the distributed storage setting, this allows the information on a failed node to be recovered with the information stored on the remaining nodes. In particular, a codeword is stored so that each node stores a symbol and recovering a failed node exactly is equivalent to fixing an erasure in the codeword \cite{DSS1}, \cite{DSS2}. 

An {\it evaluation code}~\cite{JVV} may be defined by sets of evaluation points and polynomials. Every codeword coordinate of an evaluation code depends of one of the evaluation points.  \textit{Monomial-Cartesian codes}($\mcart$)~\cite{mcc} are evaluation codes that allow for more finely-tuned polynomial sets than $\rs$ and $\grm$ codes which employ polynomials of restricted degrees.  {\it Decreasing monomial-Cartesian codes} ($\dmcc$) are a particular case of $\mcart$ that satisfy the property that the polynomials sets are closed under divisibility. Recently, it was shown in~\cite{CLMS} that polar codes can be seen in terms of $\dmcc.$ This more general setting provides the opportunity to design high rate evaluation codes that admits a repair scheme, complementing the work done for Reed-Muller codes \cite{RM}. We will see these new codes compare favorably with existing families. 

In particular, we introduce {\it augmented Reed-Muller} (\arm) codes and {\it augmented Cartesian} (\arm) codes via monomial-Cartesian codes. These augmented codes are evaluation codes obtained when certain vectors are added to a $\grm$ code and a Cartesian code, respectively. Thus the dimension is increased. We develop repair schemes for one or two erasures for these families of augmented codes. Unlike the repair scheme for two erasures of $\dmcc,$ the repair scheme for two erasures for the augmented codes has no restrictions about the positions of the erasures. Because the GW-scheme repairs a $\rs$ code provided the code satisfies a restriction on the dimension, there are codes and parameters for which the GW-scheme does not apply. In this paper, we fill some of those gaps using $\arm$ codes. When the dimension and base field are fixed, there are instances where $\arm$ codes provide a lower bandwidth in comparison with $\rs$ codes and a lower bitwidth versus Hermitian codes. When the length and base field are fixed, we give examples where $\acart$ codes provide a lower bandwidth in comparison with $\arm.$

In Section~\ref{pre}, we provide notation and definitions needed for the rest of the work. This section includes the necessary background on the families and the main properties of codes for which we develop repair schemes: decreasing monomial-Cartesian, augmented Reed-Muller, and augmented Cartesian codes. In Sections~\ref{single} and \ref{two}, we develop repair schemes for one and two erasures, respectively, on the families $\dmcc$ (with some restrictions on the positions of the erasures), $\arm$ and $\acart.$ In Section~\ref{results}, we explain some circumstances where a particular family may be preferable to others. Section~\ref{conclusions} concludes the paper with a summary of the main ideas and results of the work.

\section{Preliminaries}\label{pre}
Let $q$ be a power of a prime $p$, $\mathbb{F}_q$ denote the finite field with $q$ elements, and $K=\mathbb{F}_{q^t}$ be an  extension field of $\F_q$ of degree $t = [K : \mathbb{F}_q].$ Given a linear code $C$ of length $n$ over $K$, elements of the field extension $K$ are called {\it symbols} and the elements of the base field $\mathbb{F}_q$ are called {\it subsymbols}. As $K$ is an $\mathbb{F}_q$-vector space, every coordinate for every vector $c\in C$ depends of $t$ subsymbols. A {\it repair scheme} is an algorithm that recovers any component $c_i$ of the vector $c\in C$ using other components. The {\it bandwidth} $b$ is the number of subsymbols that the scheme needs to download to recover an erased entry $c_i$. As a vector $c \in K^n$ is composed of $nt$ subsymbols, the {\it bandwidth rate} $\displaystyle \frac{b}{nt}$ represents the fraction of the vector $c$ that the repair scheme uses to recover the erased entry $c_i$. The {\it bitwidth} $b\log_2(q)$ represents the number of bits that the scheme needs to download to recover the erased entry $c_i$.

 The {\it field trace} is defined as the polynomial $\tr_{K/\mathbb{F}_q}(x) \in K[x]$ given by
\[\tr_{K/\mathbb{F}_q}(x) = x^{q^{t-1}} + \dots + x^{q^0}.\]
For the sake of convenience, we will often refer to $\tr_{K/\mathbb{F}_q}(x)$ as simply $\tr(x)$ when the extension being used is obvious from the context. Given $a \in K$, the field trace $\tr(a) \in \mathbb{F}_q$. Additionally, $\tr: K \rightarrow \F_q$ is an $\mathbb{F}_q$-linear map. 
More useful properties of the trace function are found in Remarks \ref{Dual Bases Remark} and \ref{Trace Ker} below. They will be necessary for the repair schemes for decreasing and augmented codes.

\begin{remark}\cite[Definition 2.30 and Theorem 2.40]{Finite Fields Book}\label{Dual Bases Remark}
Let $\mathcal B=\{z_1, \dots, z_t\}$ be a basis of $K$ over $\mathbb{F}_q.$ Then there exists a basis  $\{z^\prime_1, \dots, z^\prime_t\}$ of $K$ over $\mathbb{F}_q$, called the dual basis of $\mathcal B$ such that $\tr(z_i z^\prime_j)=\delta_{ij}$ is an indicator function. For $\alpha \in K$, \[\alpha = \sum_{i=1}^t \tr(\alpha z_i)z_i^\prime.\]
Thus, determining $\alpha$ is equivalent to finding $\tr(\alpha z_i),$ for $i \in \{ 1, \dots, t \}$.
\end{remark}
The next observation follows directly from the Rank-Nullity Theorem. 
\begin{remark}\label{Trace Ker}
Given $\alpha \in K \setminus \{ 0 \}$, consider $\tr(\alpha x)$ as a function of $x.$ Then $\ker(\tr(\alpha x))$ has dimension $t - 1$ as an $\F_q$-vector space. 
\end{remark}
Next, we review decreasing monomial-Cartesian codes, setting the foundation for the augmented codes. \rmv{We continue with the same notation than previous section, in particular, $q$ represents a power of a prime $p$, $\mathbb{F}_q$ is a finite field with $q$ elements and $K=\mathbb{F}_{q^t}$ a field extension of degree $t = [K : \mathbb{F}_q].$} Let $R = K[x_1, \dots, x_m]$ be the set of polynomials in $m$ variables over $K.$ For a lattice point $\bm{a} = (a_1, \dots, a_m) \in \zm$, $\bm{x}^{\bm{a}}$ denotes the monomial $x_1^{a_1}\cdots x_m^{a_m} \in R$. For $\ell \in \mathbb{Z}_{\geq 0},$ $[\ell]:=\{1,\ldots, \ell \}.$ Given a finite set $\mathcal{A} \subset \zm$, 
the subspace of polynomials of $R$ that are $K$-linear combinations of monomials $\bm{x}^{\bm{a}},$ where $\bm{a} \in \mathcal{A}$, is \[ \polyset=\operatorname{Span}_K\{\bm{x}^{\bm{a}} : \bm{a} \in \mathcal{A} \}\subseteq R.\]

Let $\mathcal{S}=S_1\times\cdots\times S_m \subseteq K^m$ be a Cartesian product, where every $S_i \subseteq K$ has $n_i:=|S_i|>0,$ and $n:=|\mathcal{S}|$.  We will assume that $n_1\leq \cdots \leq n_m.$ Fix a linear order on $\mathcal{S}=\{\bm{s}_1,\ldots,\bm{s}_n\},$ $\bm{s}_1 \prec \cdots \prec \bm{s}_n.$ The \textit{monomial-Cartesian code} associated with $\mathcal{S}$ and $\mathcal{A}$ is given by
\[\mcc = \left\{ \ev_{\mathcal{S}}(f) : f \in \polyset \right\} \subseteq K^n,\]
where $\ev_{\mathcal{S}}(f) = \left( f(\bm{s}_1),\ldots,f(\bm{s}_n)\right).$
From now on, we assume that the degree of each monomial $M\in \polyset$ in $x_i$ is less than $n_i$. Then 
the length and rate of the monomial-Cartesian code $\mcc$ are given by $|\mathcal{S}|$ and $\frac{|\mathcal{A}|}{|\mathcal{S}|},$ respectively \cite[Proposition 2.1]{mcc}.  

The dual of $\mcc$, denoted by $\mcc^\perp$, is the set of all $\bm{\alpha} \in K^n$ such that $\bm{\alpha} \cdot \bm{\beta}=0$ for all $\bm{\beta}\in\mcc$, where $\bm{\alpha} \cdot \bm{\beta}$ is the ordinary inner product in $K^n$. The dual code $\mcc^\perp$ was studied in \cite{mcc} in terms of the vanishing ideal of $\mathcal{S}$ and in \cite{LSV} in terms of the indicator functions of $\mathcal{S}.$

It is useful to focus on the case where the monomial set $\polyset$ is {\it closed under divisibility}, meaning $\polyset$ satisfies the property that if $M\in \polyset$ and $M^\prime$ divides $M,$ then $M^\prime \in \polyset.$ In this case, the code $\mcc$ is called a {\it decreasing monomial-Cartesian} code. According to \cite[Theorem 3.3]{CLMS}\label{21.05.14}, the dual of a decreasing monomial-Cartesian code is also a decreasing monomial-Cartesian code: 
\[C(\mathcal{S},\mathcal{A})^{\perp} = \operatorname{Span}_K\left\{\operatorname{Res}_{\mathcal{S}}\left( {\bm{x}^{\bm{a}}}\right): \bm{a}\in \mathcal{A}^\complement_{\mathcal{S}}\right\} \] where 
$\displaystyle \operatorname{Res}_{\mathcal S}(f)=
\left(\lambda_{\bm{s}_1}f(\bm{s}_1),\ldots, \lambda_{\bm{s}_n}f(\bm{s}_n) \right),$
$\displaystyle \lambda_{\bm s}=\left(
\prod_{i=1}^m \prod_{\substack{s_{i}^\prime\in S_i\setminus\{s_i\}}}\left(s_i-s_i^{\prime}\right)
\right)^{-1}$, and $$\mathcal{A}^\complement_{\mathcal{S}}=\left\{ (n_1-1-a_1, \dots,n_m-1- a_m) \in \zm : \bm{a} \notin \mathcal{A} \right\}.$$
In fact, taking $\operatorname{D}_{\mathcal{S}} = \operatorname{diag}\left(\lambda_{\bm{s}_1},\ldots,\lambda_{\bm{s}_n} \right)$ to be the $n \times n$ diagonal matrix with $\lambda_{\bm{s}_i}$ in position $(i,i)$ and $0$ in any other position, it is immediate that 
\begin{equation} \label{21.05.23}
C(\mathcal{S},\mathcal{A})^{\perp} = \operatorname{D}_{\mathcal{S}} C(\mathcal{S},\mathcal{A}_{\mathcal{S}}^\complement).
\end{equation}

A 
 {\it Cartesian code}, introduced in \cite{Geil} and then independently in \cite{lopez-villa}, is defined by
\[ \car(\mathcal{S}, k) = \mathcal{C}(\mathcal{S}, \mathcal{A}_{Car}(k)), \]
where $\mathcal{A}_{Car}(k) = \{ \bm{a} \in \zm : a_i \leq n_i-1, \sum_{i=1}^m a_i \leq k \}$. 
By  Equation~(\ref{21.05.23}), the dual of the Cartesian code $\car(\mathcal{S}, k)$ is given by
$\car(\mathcal{S}, k)^\perp=\operatorname{D}_{\mathcal{S}} \car(\mathcal{S}, k^\perp),$ where $k^\perp = \sum_{i=1}^m (n_i-1) - k - 1$.

Observe that if $\mathcal{S}=K^m,$ the Cartesian code $\car(\mathcal{S}, k)$ is the Reed-Muller code $\grm(K^m, k).$ The dual code $\grm(K^m, k)^\perp$ has been extensively studied in the literature. See for instance~\cite{RM,DGM,Huffman-Pless}, where it is shown that the dual of a $\grm$ code is another $\grm$ code.

\section{Augmented codes}

In this section, we define the augmented Cartesian codes for which we will provide repair schemes in the following section. Augmented Cartesian codes generalize the augmented Reed-Muller codes considered in \cite{LMV}. Keeping the notation from the previous sections, we describe two families below.

\subsection{Augmented Cartesian codes 1}
An {\it augmented Cartesian code 1} (ACar1 code) over $K=\F_{q^t}$ is defined by
\[ \acar(\mathcal{S}, {\bm k}) = \mathcal{C}(\mathcal{S}, \mathcal{A}_{Car1}({\bm k})), \]
where $\bm{k} = (k_1, \dots, k_m),$ with $ 0 \leq k_i \leq n_i - q^{t - 1},$ and
\[\mathcal{A}_{Car1}(\bm{k}) = \prod_{i = 1}^m \left\{0,\ldots,n_i-1\right\} \setminus \prod_{i = 1}^m\left\{k_i,\ldots,n_i-1\right\} .\]
An augmented Cartesian code 1 is shown in Example~\ref{21.01.01}.
\begin{example}\label{21.01.01} \rm
Take $K=\mathbb{F}_{17}$. Let $S_1, S_2 \subseteq K$ with $n_1 = |S_1| = 6$ and $n_2 = |S_2| = 7$. The code $\acar(S_1 \times S_2, (2,2))$ is generated by the vectors $\ev_{S_1 \times S_2} (\text{\textcolor{red}{$M$}}),$ where \textcolor{red}{$M$} is a monomial whose exponent is a point in Figure~\ref{ExACar1} (a). The dual code $\acar(S_1 \times S_2, (2,2))^{\perp}$ is generated by the vectors $\operatorname{Res}_{S_1 \times S_2} (\text{\textcolor{blue}{$M$}}),$ where \textcolor{blue}{$M$} is a
monomial whose exponent is a point in Figure~\ref{ExACar1} (b).
\begin{figure}[h]
\noindent
\begin{minipage}[t]{0.23\textwidth}
\begin{tikzpicture}[scale=0.55]
\draw [-latex] (0,0) -- (5.5,0)node[right] {$A_1$};
\draw [dashed] (0,1)node[left]{1} -- (5,1)node[right] {};
\draw [dashed] (0,2)node[left]{2} -- (5,2)node[right] {};
\draw [dashed] (0,3)node[left]{3} -- (5,3)node[right] {};
\draw [dashed] (0,4)node[left]{4} -- (5,4)node[right] {};
\draw [dashed] (0,5)node[left]{5} -- (5,5)node[right] {};
\draw [dashed] (0,6)node[left]{6} -- (5, 6)node[right] {};
\draw [-latex] (0,0)node[below left]{0} -- (0,6.5)node[above] {$A_2$};
\draw [dashed] (1,0)node[below]{1} -- (1,6)node[right] {};
\draw [dashed] (2,0)node[below]{2} -- (2,6)node[right] {};
\draw [dashed] (3,0)node[below]{3} -- (3,6)node[right] {};
\draw [dashed] (4,0)node[below]{4} -- (4,6)node[right] {};
\draw [dashed] (5,0)node[below]{5} -- (5,6)node[right] {};

\foreach \i in {0,...,5}
{\foreach \j in {0,...,1}
{\fill [color=red](\i,\j) {circle(.2cm)};}}
\foreach \i in {0,...,1}
{\foreach \j in {2,...,6}
{\fill [color=red](\i,\j) {circle(.2cm)};}}
\end{tikzpicture}
\begin{center}
(a)
\end{center}
\end{minipage}
\begin{minipage}[t]{0.23\textwidth}
\begin{tikzpicture}[scale=0.55]
\draw [-latex] (6,6) -- (6,-0.5)node[below] {$A_2$};
\draw [dashed] (1,0)node[left]{} -- (6,0)node[right] {6};
\draw [dashed] (1,1)node[left]{} -- (6,1)node[right] {5};
\draw [dashed] (1,2)node[left]{} -- (6,2)node[right] {4};
\draw [dashed] (1,3)node[left]{} -- (6,3)node[right] {3};
\draw [dashed] (1,4)node[left]{} -- (6,4)node[right] {2};
\draw [dashed] (1,5)node[left]{} -- (6,5)node[right] {1};
\draw [dashed] (1,6)node[left]{} -- (6,6)node[right] {};
\draw [-latex] (6,6)node[above right]{0} -- (0.5,6)node[left] {$A_1$};
\draw [dashed] (1,0)node[below]{} -- (1,6)node[above] {5};
\draw [dashed] (2,0)node[below]{} -- (2,6)node[above] {4};
\draw [dashed] (3,0)node[below]{} -- (3,6)node[above] {3};
\draw [dashed] (4,0)node[below]{} -- (4,6)node[above] {2};
\draw [dashed] (5,0)node[below]{} -- (5,6)node[above] {1};
\draw [dashed] (6,0)node[below]{} -- (6,6)node[above] {};
\foreach \i in {3,...,6}
{\foreach \j in {2,...,6}
{\fill [color=blue](\i,\j) {circle(.2cm)};}}
\end{tikzpicture}
\begin{center}
(b)
\end{center}
\end{minipage}
\caption{ The code $\acar(S_1 \times S_2, (2,2))$ in Example \ref{21.01.01} with $K=\mathbb{F}_{17}$ is generated by the vectors $\ev_{S_1 \times S_2} (\text{\textcolor{red}{$M$}}),$ where \textcolor{red}{$M$} is a monomial whose exponent is a point in (a). The dual code $\acar(S_1 \times S_2, (2,2))^{\perp}$ is generated by the vectors $\operatorname{Res}_{S_1 \times S_2} (\text{\textcolor{blue}{$M$}}),$ where \textcolor{blue}{$M$} is a monomial whose exponent is a point in (b).}
\label{ExACar1}
\end{figure}
\end{example}

When  $k_i=k \leq q^t-q^{t-1}$ for all $i \in [m]$ and $\mathcal{S}=K^m,$ the augmented Cartesian code 1 $\acar(\mathcal{S}, {\bm k})$ is called an {\it augmented Reed-Muller code 1}, which is denoted by $\agrm(K^m, k).$ An augmented Reed-Muller code 1 is shown in Example~\ref{21.05.16}.
\begin{example}\label{21.05.16} \rm
Take $K=\mathbb{F}_7.$ The code $\agrm(K^2, 2)$ is generated by the vectors
$\ev_{K^2} (\text{\textcolor{red}{$M$}}),$ where \textcolor{red}{$M$} is a monomial
whose exponent is a point in Figure~\ref{ExARM1} (a). The dual
$\agrm(K^2, 2)^{\perp}$ is generated by the vectors
$\ev_{K^2} (\text{\textcolor{blue}{$M$}}),$ where \textcolor{blue}{$M$} is a
monomial whose exponent is a point in Figure~\ref{ExARM1} (b).
\begin{figure}[h]
\vskip 0cm
\noindent
\begin{minipage}[t]{0.45\textwidth}
\begin{center}
\begin{tikzpicture}[scale=0.55]
\draw [-latex] (0,0) -- (6.5,0)node[right] {$K$};
\draw [dashed] (0,1)node[left]{1} -- (6,1)node[right] {};
\draw [dashed] (0,2)node[left]{2} -- (6,2)node[right] {};
\draw [dashed] (0,3)node[left]{3} -- (6,3)node[right] {};
\draw [dashed] (0,4)node[left]{4} -- (6,4)node[right] {};
\draw [dashed] (0,5)node[left]{5} -- (6,5)node[right] {};
\draw [dashed] (0,6)node[left]{6} -- (6,6)node[right] {};
\draw [-latex] (0,0)node[below left]{0} -- (0,6.5)node[above] {$K$};
\draw [dashed] (1,0)node[below]{1} -- (1,6)node[right] {};
\draw [dashed] (2,0)node[below]{2} -- (2,6)node[right] {};
\draw [dashed] (3,0)node[below]{3} -- (3,6)node[right] {};
\draw [dashed] (4,0)node[below]{4} -- (4,6)node[right] {};
\draw [dashed] (5,0)node[below]{5} -- (5,6)node[right] {};
\draw [dashed] (6,0)node[below]{6} -- (6,6)node[right] {};
\foreach \i in {0,...,6}
{\foreach \j in {0,...,1}
{\fill [color=red](\i,\j) {circle(.2cm)};}}
\foreach \i in {0,...,1}
{\foreach \j in {2,...,6}
{\fill [color=red](\i,\j) {circle(.2cm)};}}
\end{tikzpicture}
\vskip 0cm
(a)
\end{center}
\end{minipage}
\begin{minipage}[t]{0.4\textwidth}
\begin{center}
\begin{tikzpicture}[scale=0.55]
\draw [-latex] (6,6) -- (6,-0.5)node[below] {$K$};
\draw [dashed] (0,0)node[left]{} -- (6,0)node[right] {6};
\draw [dashed] (0,1)node[left]{} -- (6,1)node[right] {5};
\draw [dashed] (0,2)node[left]{} -- (6,2)node[right] {4};
\draw [dashed] (0,3)node[left]{} -- (6,3)node[right] {3};
\draw [dashed] (0,4)node[left]{} -- (6,4)node[right] {2};
\draw [dashed] (0,5)node[left]{} -- (6,5)node[right] {1};
\draw [dashed] (0,6)node[left]{} -- (6,6)node[right] {};
\draw [-latex] (6,6)node[above right]{0} -- (-0.5,6)node[left] {$K$};
\draw [dashed] (0,0)node[below]{} -- (0,6)node[above] {6};
\draw [dashed] (1,0)node[below]{} -- (1,6)node[above] {5};
\draw [dashed] (2,0)node[below]{} -- (2,6)node[above] {4};
\draw [dashed] (3,0)node[below]{} -- (3,6)node[above] {3};
\draw [dashed] (4,0)node[below]{} -- (4,6)node[above] {2};
\draw [dashed] (5,0)node[below]{} -- (5,6)node[above] {1};
\draw [dashed] (6,0)node[below]{} -- (6,6)node[above] {};
\foreach \i in {2,...,6}
{\foreach \j in {2,...,6}
{\fill [color=blue](\i,\j) {circle(.2cm)};}}
\end{tikzpicture}
\vskip 0.cm
(b)
\end{center}
\end{minipage}
\caption{ The $\agrm(K^2, 2)$ code in Example \ref{21.05.16} with $K=\mathbb{F}_7$ is generated by the vectors 
$\ev_{K^2} (\text{\textcolor{red}{$M$}}),$ where \textcolor{red}{$M$} is a monomial
whose exponent corresponds to a point in (a).
$\agrm(K^2, 2)^{\perp}$ is generated by the vectors
$\ev_{K^2} (\text{\textcolor{blue}{$M$}}),$ where \textcolor{blue}{$M$} is a monomial whose exponent corresponds to a point in (b).}
\label{ExARM1}
\vskip 0.cm
\end{figure}
\end{example}
In Figure~\ref{ExARM1}, the monomials that define $\grm(K^2, 2)$ may be seen as those under the diagonal in $\agrm(K^2, 2)$.  The monomial diagram for any Reed-Muller code will restrict the allowable monomials under some diagonal excluding many monomials along or near the edges, resulting in codes with lower dimensions and rates. This explains why $\agrm$ codes have higher rates than their associated Reed-Muller codes.

The next result is relevant for developing the repair scheme for $\acar(\mathcal{S}, {\bm k}).$ 
\begin{proposition}\label{21.05.26}
The following holds for the augmented Cartesian code 1.
\begin{itemize}
\item[\rm (a)] The dimension is $\displaystyle \dim \acar(\mathcal{S}, {\bm k}) = \prod_{j = 1}^m n_j - \prod_{j = 1}^m (n_j - k_j).$
\item[\rm (b)] The dual is $\displaystyle \acar(\mathcal{S}, {\bm k})^\perp = \operatorname{D}_{\mathcal{S}} \mathcal{C}(\mathcal{S}, \mathcal{A}_{Car1}^\perp(\bm{k}) ),$\newline
where $\displaystyle \mathcal{A}_{Car1}^\perp(\bm{k})= \prod_{i=1}^m\left\{0,\ldots,n_i-k_i-1\right\}.$
\end{itemize}
\end{proposition}
\begin{proof}
(a) The statement follows immediately, because
\[ \displaystyle |\mathcal{A}_{Car1}(\bm{k})| = |\prod_{j = 1}^m \{0, \dots, n_j - 1\} \setminus \prod_{j = 1}^m \{k_j , \dots, n_j - 1\}| = \displaystyle \prod_{j = 1}^m n_j - \prod_{j = 1}^m (n_j - k_j ).\] \newline
(b) Observe that
$\displaystyle \mathcal{A}_{Car1} ({\bm k})^\complement_{\mathcal{S}} = \mathcal{A}_{Car1}^\perp(\bm{k}).$
Indeed, \newline $\displaystyle \left(n_1-1-a_1,\ldots,n_m-1-a_m\right)\in \mathcal{A}_{Car1} (\bm{k})^\complement_{\mathcal{S}}$
if and only if \newline
$\displaystyle (a_1,\ldots,a_m)\in \prod_{i = 1}^m \left\{0,\ldots,n_i-1\right\} \setminus \mathcal{A}_{Car1}(\bm{k}),$
which happens if and only if \newline
$\displaystyle (n_1-1,\ldots,n_m-1)-(a_1,\ldots,a_m)\in\mathcal{A}_{Car1}^\perp(\bm{k}).$ Thus, the result follows by  Equation~(\ref{21.05.23}).
\end{proof}
\subsection{Augmented Cartesian codes 2}
We next define a second family of high-rate Cartesian codes.
The {\it augmented Cartesian code 2} (ACar2 code) is defined by
\[ \acarr(\mathcal{S}, {\bm k}) = \mathcal{C}(\mathcal{S}, \mathcal{A}_{Car2}({\bm k})), \]
where $\bm{k} = (k_1, \dots, k_m),$ with $ 0 \leq k_i \leq n_i - q^{t - 1},$ and \newline
$\displaystyle \mathcal{A}_{Car2}(\bm{k}) = \prod_{j = 1}^m \left\{0,\ldots,n_j-1\right\} \setminus \bigcup_{j = 1}^mL_j,$ with
\[ L_j = \{ \bm{a} : k_j \leq a_j \leq n_j {-} 1 \ , \ a_i = n_i {-} 1 \text{ for all } i \neq j\}.\]
An augmented Cartesian code 2 is shown in Example~\ref{21.05.24}.
\begin{example}\label{21.05.24} \rm
Take $K=\mathbb{F}_{17}$. Let $S_1$ and $S_2$ be subsets of $K$ with $n_1 = |S_1| = 6$ and $n_2 = |S_2| = 7$.
The code $\acarr(S_1 \times S_2, (2, 5))$ is generated by the vectors $\ev_{S_1 \times S_2} (\text{\textcolor{red}{$M$}}),$
where \textcolor{red}{$M$} is a monomial whose exponent is a point in Figure~\ref{21.05.21} (a). The dual code
$\acarr(S_1 \times S_2, (2, 5))^{\perp}$ is generated by the vectors
$\operatorname{Res}_{S_1 \times S_2} (\text{\textcolor{blue}{$M$}}),$ where \textcolor{blue}{$M$} is a
monomial whose exponent is a point in Figure~\ref{21.05.21} (b).
\begin{figure}[h]
\noindent
\begin{minipage}[t]{0.23\textwidth}
\begin{tikzpicture}[scale=0.55]
\draw [-latex] (0,0) -- (5.5,0)node[right] {$A_1$};
\draw [dashed] (0,1)node[left]{1} -- (5,1)node[right] {};
\draw [dashed] (0,2)node[left]{2} -- (5,2)node[right] {};
\draw [dashed] (0,3)node[left]{3} -- (5,3)node[right] {};
\draw [dashed] (0,4)node[left]{4} -- (5,4)node[right] {};
\draw [dashed] (0,5)node[left]{5} -- (5,5)node[right] {};
\draw [dashed] (0,6)node[left]{6} -- (5,6)node[right] {};
\draw [-latex] (0,0)node[below left]{0} -- (0,6.5)node[above] {$A_2$};
\draw [dashed] (1,0)node[below]{1} -- (1,6)node[right] {};
\draw [dashed] (2,0)node[below]{2} -- (2,6)node[right] {};
\draw [dashed] (3,0)node[below]{3} -- (3,6)node[right] {};
\draw [dashed] (4,0)node[below]{4} -- (4,6)node[right] {};
\draw [dashed] (5,0)node[below]{5} -- (5,6)node[right] {};
\foreach \i in {0,...,4}
{\foreach \j in {0,...,5}
{\fill [color=red](\i,\j) {circle(.2cm)};}}
\foreach \i in {0,...,1}
{\fill [color=red](\i,6) {circle(.2cm)};}
\foreach \j in {0,...,4}
{\fill [color=red](5,\j) {circle(.2cm)};}
\end{tikzpicture}
\begin{center}
(a)
\end{center}
\end{minipage}
\begin{minipage}[t]{0.23\textwidth}
\begin{tikzpicture}[scale=0.55]
\draw [-latex] (6,6) -- (6,-0.5)node[below] {$A_2$};
\draw [dashed] (1,0)node[left]{} -- (6,0)node[right] {6};
\draw [dashed] (1,1)node[left]{} -- (6,1)node[right] {5};
\draw [dashed] (1,2)node[left]{} -- (6,2)node[right] {4};
\draw [dashed] (1,3)node[left]{} -- (6,3)node[right] {3};
\draw [dashed] (1,4)node[left]{} -- (6,4)node[right] {2};
\draw [dashed] (1,5)node[left]{} -- (6,5)node[right] {1};
\draw [dashed] (1,6)node[left]{} -- (6,6)node[right] {};
\draw [-latex] (6,6)node[above right]{0} -- (0.5,6)node[left] {$A_1$};
\draw [dashed] (1,0)node[below]{} -- (1,6)node[above] {5};
\draw [dashed] (2,0)node[below]{} -- (2,6)node[above] {4};
\draw [dashed] (3,0)node[below]{} -- (3,6)node[above] {3};
\draw [dashed] (4,0)node[below]{} -- (4,6)node[above] {2};
\draw [dashed] (5,0)node[below]{} -- (5,6)node[above] {1};
\draw [dashed] (6,0)node[below]{} -- (6,6)node[above] {};
\foreach \i in {3,...,6}
{\fill [color=blue](\i,6) {circle(.2cm)};}
\foreach \j in {5,...,6}
{\fill [color=blue](6,\j) {circle(.2cm)};}
\end{tikzpicture}
\begin{center}
(b)
\end{center}
\end{minipage}
\caption{ The code $\acarr(S_1 \times S_2, (2, 5))$ in Example \ref{21.05.24} with $K=\mathbb{F}_{17}$ is generated by the vectors 
$\ev_{S_1 \times S_2} (\text{\textcolor{red}{$M$}}),$
where \textcolor{red}{$M$} is a monomial whose exponent is a point in (a). The dual
$\acarr(S_1 \times S_2, (2, 5))^{\perp}$ is generated by the vectors
$\operatorname{Res}_{S_1 \times S_2} (\text{\textcolor{blue}{$M$}}),$ where \textcolor{blue}{$M$} is a
monomial whose exponent is a point in (b).}
\label{21.05.21}
\end{figure}
\end{example}

When $k_i=k\leq q^t-q^{t-1}$ for all $i \in [m]$ and $\mathcal{S}=K^m,$ the augmented Cartesian code 2 $\acarr(\mathcal{S}, {\bm k})$ is called an {\it augmented Reed-Muller code} 2, which is denoted by $\agrmm(K^m, k).$ An augmented Reed-Muller code 2 is shown in Example~\ref{21.01.03}.
\begin{example}\label{21.01.03} \rm
Take $K=\mathbb{F}_7.$ The code $\agrmm(K^2, 3)$ is generated by the vectors
$\ev_{K^2} (\text{\textcolor{red}{$M$}}),$ where \textcolor{red}{$M$} is a monomial
whose exponent is a point in Figure~\ref{21.05.19} (a). The dual
$\agrmm(K^2, 3)^{\perp}$ is generated by the vectors
$\ev_{K^2} (\text{\textcolor{blue}{$M$}}),$ where \textcolor{blue}{$M$} is a
monomial whose exponent is a point in Figure~\ref{21.05.19} (b).
\begin{figure}[h]
\vskip 0.cm
\noindent
\begin{minipage}[t]{0.23\textwidth}
\begin{tikzpicture}[scale=0.55]
\draw [-latex] (0,0) -- (6.5,0)node[right] {$K$};
\draw [dashed] (0,1)node[left]{1} -- (6,1)node[right] {};
\draw [dashed] (0,2)node[left]{2} -- (6,2)node[right] {};
\draw [dashed] (0,3)node[left]{3} -- (6,3)node[right] {};
\draw [dashed] (0,4)node[left]{4} -- (6,4)node[right] {};
\draw [dashed] (0,5)node[left]{5} -- (6,5)node[right] {};
\draw [dashed] (0,6)node[left]{6} -- (6,6)node[right] {};
\draw [-latex] (0,0)node[below left]{0} -- (0,6.5)node[above] {$K$};
\draw [dashed] (1,0)node[below]{1} -- (1,6)node[right] {};
\draw [dashed] (2,0)node[below]{2} -- (2,6)node[right] {};
\draw [dashed] (3,0)node[below]{3} -- (3,6)node[right] {};
\draw [dashed] (4,0)node[below]{4} -- (4,6)node[right] {};
\draw [dashed] (5,0)node[below]{5} -- (5,6)node[right] {};
\draw [dashed] (6,0)node[below]{6} -- (6,6)node[right] {};
\foreach \i in {0,...,5}
{\foreach \j in {0,...,5}
{\fill [color=red](\i,\j) {circle(.2cm)};}}
\foreach \i in {0,...,2}
{\fill [color=red](\i,6) {circle(.2cm)};}
\foreach \j in {0,...,2}
{\fill [color=red](6,\j) {circle(.2cm)};}
\end{tikzpicture}
\begin{center}
\vskip 0.cm
(a)
\end{center}
\end{minipage}
\begin{minipage}[t]{0.23\textwidth}
\begin{tikzpicture}[scale=0.55]
\draw [-latex] (6,6) -- (6,-0.5)node[below] {$K$};
\draw [dashed] (0,0)node[left]{} -- (6,0)node[right] {6};
\draw [dashed] (0,1)node[left]{} -- (6,1)node[right] {5};
\draw [dashed] (0,2)node[left]{} -- (6,2)node[right] {4};
\draw [dashed] (0,3)node[left]{} -- (6,3)node[right] {3};
\draw [dashed] (0,4)node[left]{} -- (6,4)node[right] {2};
\draw [dashed] (0,5)node[left]{} -- (6,5)node[right] {1};
\draw [dashed] (0,6)node[left]{} -- (6,6)node[right] {};
\draw [-latex] (6,6)node[above right]{0} -- (-0.5,6)node[left] {$K$};
\draw [dashed] (0,0)node[below]{} -- (0,6)node[above] {6};
\draw [dashed] (1,0)node[below]{} -- (1,6)node[above] {5};
\draw [dashed] (2,0)node[below]{} -- (2,6)node[above] {4};
\draw [dashed] (3,0)node[below]{} -- (3,6)node[above] {3};
\draw [dashed] (4,0)node[below]{} -- (4,6)node[above] {2};
\draw [dashed] (5,0)node[below]{} -- (5,6)node[above] {1};
\draw [dashed] (6,0)node[below]{} -- (6,6)node[above] {};
\foreach \i in {3,...,6}
{\fill [color=blue](\i,6) {circle(.2cm)};}
\foreach \j in {3,...,6}
{\fill [color=blue](6,\j) {circle(.2cm)};}
\end{tikzpicture}
\begin{center}
\vskip 0.cm
(b)
\end{center}
\end{minipage}
\caption{ The code $\agrmm(K^2, 3)$ in Example \ref{21.01.03} with $K=\mathbb{F}_7$ is generated by the vectors 
$\ev_{K^2} (\text{\textcolor{red}{$M$}}),$ where \textcolor{red}{$M$} is a monomial
whose exponent corresponds to a point in (a). The dual
$\agrmm(K^2, 3)^{\perp}$ is generated by the vectors
$\ev_{K^2} (\text{\textcolor{blue}{$M$}}),$ where \textcolor{blue}{$M$} is a monomial whose exponent corresponds to a point in (b).}
\label{21.05.19}
\vskip 0.cm
\end{figure}
\end{example}

The next result is relevant for developing the repair scheme for $\acarr(\mathcal{S}, {\bm k}).$ 
\begin{proposition}
The following holds for the augmented Cartesian code 2.
\begin{itemize}
\item[\rm (a)] The dimension is $\displaystyle \dim \acarr(\mathcal{S}, {\bm k}) = \prod_{i = 1}^m n_i - \sum_{i = 1}^m (n_i - k_i-1) -1.$
\item[\rm (b)] The dual is $\displaystyle \acarr(\mathcal{S}, {\bm k})^\perp = \operatorname{D}_{\mathcal{S}} \mathcal{C}(\mathcal{S}, \mathcal{A}_{Car2}^\perp(\bm{k}) ),$
where $\displaystyle \mathcal{A}_{Car2}^\perp(\bm{k}) = \bigcup_{j=1}^mL_j^\prime $ and 
\[ L_j^\prime = \{ \bm{a} \ : \ 0 \leq a_j \leq n_j {-} k_j {-} 1 \ , \ a_i = 0 \text{ for all } i \neq j\}.\]
\end{itemize}
\end{proposition}
\begin{proof}
We have $\displaystyle |\mathcal{A}_{Car2}(\bm{k})| = |\prod_{i = 1}^m \left\{0,\ldots,n_i-1\right\} \setminus \bigcup_{i = 1}^mL_i| = \prod_{i = 1}^m n_i - |\bigcup_{i = 1}^m L_i|$. As $\displaystyle \bigcap_{i=1}^mL_i=\{\bm{a}\},$ where
$\displaystyle \bm{a}=(n_1 - 1,\ldots, n_m - 1),$ and
$\displaystyle \left( L_i\setminus \{\bm{a}\} \right) \bigcap \left( L_j\setminus \{\bm{a}\} \right)=\emptyset$ for all $i\neq j,$ then
$\displaystyle |\bigcup_{i=1}^mL_i| = \sum_{i=1}^m |L_i\setminus \{\bm{a}\}| + 1= \sum_{i = 1}^m (n_i - k_i - 1) +1.$
Thus $\displaystyle |\mathcal{A}_{Car2}(\bm{k})| = \prod_{i = 1}^m n_i {-} \sum_{i = 1}^m (n_i - k_i - 1) - 1$. (b) In a similar way to the proof of Proposition~\ref{21.05.26}, it is not difficult to check that
$\mathcal{A}_{Car2} (\bm{k})^\complement_{K^m} = \mathcal{A}_{Car2}^\perp(\bm{k}).$ Thus, the result follows from  Equation~(\ref{21.05.23}).
\end{proof}

\section{Single Erasure Repair Schemes}\label{single}
In this section, we develop a repair scheme that repairs a single erasure of a decreasing monomial-Cartesian code $\mcc$ that satisfies the property that $\mathcal{A} \cap L_j = \emptyset$ for some $j$, where $L_j= \{ \bm{a} : n_j - q^{t - 1} \leq  a_j < n_j, a_i = n_{i-1}-1 \text{ for } i\neq j \}.$ As a consequence, we obtain repair schemes for single erasures of augmented Cartesian and Reed-Muller codes.

\rmv{
We continue with the same notation than previous section, in particular, $q$ represents a power of a prime $p$, $\mathbb{F}_q$ is a finite field with $q$ elements and $K=\mathbb{F}_{q^t}$ a field extension of degree $t = [K : \mathbb{F}_q].$ We also have that $\mathcal{S}=S_1\times\cdots\times S_m \subseteq K^m$ is a Cartesian product, $n_i$ denotes $|S_i|,$ and $n$ denotes $|\mathcal{S}|=\prod _{i=1}^m n_i.$ We are assuming that $n_1\leq \cdots \leq n_m,$ that the degree of each monomial $M\in \polyset$ in $x_i$ is less than $n_i,$ and that the monomial set $\polyset$ is closed under divisibility, in other words, $\polyset$ satisfies the property that if $M\in \polyset$ and $M^\prime$ divides $M,$ then $M^\prime \in \polyset.$}

\begin{theorem}\label{21.06.17}
Let $\mcc$ be a decreasing monomial-Cartesian code of length $n$ such that there is $j\in [n]$ with $\mathcal{A} \cap L_j = \emptyset$. Then there exists a repair scheme for one erasure with bandwidth at most
\[b = n-1 + (t-1)\left(\frac{n}{n_j}-1\right).\]
\end{theorem}
\begin{proof}
Let $\bm{s}^*=(s^*_1,\ldots,s^*_m) \in \mathcal{S}$ and assume that the entry $f(\bm{s}^*)$ of the codeword $(f(\bm{s}_1),\ldots,f(\bm{s}_n))\in \mcc$ has been erased. Let $\{z_1, \dots, z_t\}$ be a basis for $K$ over $\mathbb{F}_q.$ For $i\in [t],$ define the following polynomials
\begin{equation*}
p_i(\bm{x})=
\frac{\tr(z_i(x_j - s^*_j))}{(x_j - s^*_j)}
= {z_i} +
z_i^q(x_j - s^*_j)^{q-1}+ \cdots +
z_i^{q^{t-1}}(x_j - s^*_j)^{q^{t-1}-1}.
\end{equation*}
As $\mathcal{A} \cap L_j = \emptyset,$
$\left(L_j\right)_{\mathcal{S}}^\complement = \{(0,\ldots,0,a) : 0 \leq  a < q^{t - 1} \} \subseteq \mathcal{A}_{\mathcal{S}}^\complement.$ Thus, for $i\in [t],$
every polynomial $p_i(\bm{x})\in \mathcal{L}(\left(L_j\right)_{\mathcal{S}}^\complement) \subseteq \mathcal{L}(\mathcal{A}_{\mathcal{S}}^\complement)$ 
defines an element in $\mcc^\perp= \operatorname{D}_{\mathcal{S}} C(\mathcal{S},\mathcal{A}_{\mathcal{S}}^\complement)$. Therefore, we obtain the $t$ equations
\begin{equation}\label{21.06.15}
\lambda_{\bm{s}^*}p_{i}(\bm{s}^*)f(\bm{s}^*)= -\sum_{\mathcal{S} \setminus\{\bm{s}^*\}}
\lambda_{\bm{s}} p_{i}(\bm{s})f(\bm{s}),\quad i\in[t].
\end{equation}
As $p_i(\bm{s}^*)=z_i,$ applying the trace function to both sides of previous equations and employing the linearity of the trace function, we obtain
\[\tr \left(z_i \lambda_{\bm{s}^*} f(\bm{s}^*)\right)= 
-\sum_{\mathcal{S} \setminus\{\bm{s}^*\}}
\tr \left( \lambda_{\bm{s}} p_{i}(\bm{s})f(\bm{s}) \right),\quad i\in[t].\]
Define the set $ \Gamma = \{(s_1,\ldots, s_m)\in \mathcal{S} : s_j = s_j^* \}.$
For $ \bm{s} \in \Gamma,$ we have that $p_i(\bm{s})=z_i.$
For $ \bm{s} \in \mathcal{S} \setminus \Gamma,$ $p_i(\bm{s})=\displaystyle \frac{ \tr(z_i(s_j - s^*_j))}{(s_j - s^*_j)}.$ Therefore, we obtain that for $i\in[t],$
\begin{eqnarray*}
\sum_{\mathcal{S} \setminus\{\bm{s}^*\}} \tr \left( \lambda_{\bm{s}} p_{i}(\bm{s})f(\bm{s}) \right)
&=& \sum_{\Gamma \setminus\{\bm{s}^*\}} \tr \left( \lambda_{\bm{s}} p_{i}(\bm{s})f(\bm{s}) \right) + 
\sum_{\mathcal{S} \setminus \Gamma} \tr \left( \lambda_{\bm{s}} p_{i}(\bm{s})f(\bm{s}) \right) \\
&=& \sum_{\Gamma \setminus\{\bm{s}^*\}} \tr \left( \lambda_{\bm{s}} z_i f(\bm{s}) \right) + 
\sum_{\mathcal{S} \setminus \Gamma} \tr \left( \lambda_{\bm{s}} \frac{\tr(z_i(s_j - s^*_j))}{(s_j - s^*_j)}f(\bm{s}) \right)\\
&=& \sum_{\Gamma \setminus\{\bm{s}^*\}} \tr \left( \lambda_{\bm{s}} z_i f(\bm{s}) \right) + 
\sum_{\mathcal{S} \setminus \Gamma} \tr(z_i(s_j - s^*_j)) \tr \left( \frac{\lambda_{\bm{s}} f(\bm{s})}{(s_j - s^*_j)} \right).
\end{eqnarray*}
By Remark~\ref{Dual Bases Remark}, $\lambda_{\bm{s}^*} f(\bm{s}^*),$ and as a consequence, $f(\bm{s}^*)$, can be recovered from its $t$ independent traces $\tr (z_i \lambda_{\bm{s}^*} f(\bm{s}^*)),$ which can be obtained by downloading:
\begin{itemize}
\item  $t$ subsymbols $\tr\left(\lambda_{\bm{s}}z_1 f(\bm{s})\right), \ldots, \tr\left(\lambda_{\bm{s}}z_t f(\bm{s})\right),$ for each $\bm{s}\in \Gamma\setminus\{\bm{s}^*\}.$
\item subsymbol $\displaystyle \tr\left(\frac{\lambda_{\bm{s}} f(\bm{s})}{ (s_j - s^*_j)}\right),$ for each $\bm{s}\in \mathcal{S}\setminus \Gamma.$
\end{itemize}
Hence, the bandwidth is $\displaystyle b = t (|\Gamma|-1) + |\mathcal{S}\setminus \Gamma|=
(t-1)\left(\frac{n}{n_j}-1\right) + n-1.$
\end{proof}

\begin{corollary} \label{21.05.25}
There exist repair schemes for one erasure of $\acar(\mathcal{S},\bm{k})$ and $\acarr(\mathcal{S},\bm{k}),$ each with bandwidth at most
\[b = \prod_{i = 1}^m n_i - 1 + (t-1) \left(\prod_{i = 1}^{m-1} n_i - 1\right).\]
\end{corollary}
\begin{proof}
Since $k_i \leq n_i - q^{t - 1}$ for $i\in [m]$, $\mathcal{A}_{Car1}({\bm k})\cap L_m = \emptyset$ and $\mathcal{A}_{Car2}({\bm k})\cap L_m = \emptyset$, where
$L_m= \{(n_1-1,\ldots,n_{m-1}-1,a) : n_m - q^{t - 1} \leq  a < n_m \}.$ Thus, the result follows from Theorem~\ref{21.06.17}. 
\end{proof}
As another consequence from Theorem~\ref{21.06.17}, by taking $\mathcal{S} = K^m$, we obtain a repair scheme for augmented Reed-Muller codes, whose family was first introduced in~\cite[Theorem 2.5]{LMV}.
\begin{corollary}\label{21.01.07}
There exists a repair scheme for one erasure for $\agrm(K^m,k)$ and for $\agrmm(K^m,k),$ each with bandwidth
\[b = |K|^m - 1 + (t-1)(|K|^{m-1}-1).\]
\end{corollary}
\rmv{\begin{proof}
This comes from Corollary~\ref{21.05.25}, when $\mathcal{S} = K^m.$
\end{proof}}
\begin{remark}\rm\label{21.06.26}
The bandwidth of the repair scheme developed in Corollary~\ref{21.01.07} for augmented Reed-Muller codes is less than the one developed in~\cite[Theorems 2.5 and 3.4]{LMV}. This is due to the fact that the repair polynomials used in the proofs of \cite[Theorems 2.5 and 3.4]{LMV} have more zeros over $\mathcal{S}$ than the repair polynomials of the proof of Corollary~\ref{21.01.07}. Thus, the number of subsymbols that are needed to repair an erasure is less when we use Corollary~\ref{21.01.07}. \rmv{See also Remarks~\ref{21.06.27} and \ref{21.06.28}.}
\end{remark}
\section{Two Erasures Repair Schemes}\label{two}
In this section, we keep the same notation as in previous sections and develop a repair scheme that repairs two simultaneous erasures $f(\bm{s^{\prime}})$ and $f(\bm{s}^*)$ of $\mcc$ provided the erasure positions satisfy the property that $s^*_j \neq s^{\prime}_j.$ Then we give a repair scheme that repairs two simultaneous erasures of the augmented Cartesian and Reed-Muller codes that does not require that the position vectors $\bm{s^{\prime}}$ and $\bm{s}^*$ are different on a specific component.

\rmv{We continue with the same notation than previous sections, in particular, $q$ represents a power of a prime $p$, $\mathbb{F}_q$ is a finite field with $q$ elements and $K=\mathbb{F}_{q^t}$ a field extension of degree $t = [K : \mathbb{F}_q].$ We also have that $\mathcal{S}=S_1\times\cdots\times S_m \subseteq K^m$ is a Cartesian product, $n_i$ denotes $|S_i|,$ and $n$ denotes $|\mathcal{S}|=\prod _{i=1}^m n_i.$ We are assuming that $n_1\leq \cdots \leq n_m,$ that the degree of each monomial $M\in \polyset$ in $x_i$ is less than $n_i,$ and that the monomial set $\polyset$ is closed under divisibility, in other words, $\polyset$ satisfies the property that if $M\in \polyset$ and $M^\prime$ divides $M,$ then $M^\prime \in \polyset.$}

\begin{theorem}\label{21.06.16}
Let $\mcc$ be a decreasing monomial-Cartesian code of length $n$ such that there exists $j\in [n]$ with $\mathcal{A} \cap L_j = \emptyset$. Let $\bm{s}^*=(s^*_1,\ldots,s^*_m), \bm{s^{\prime}}=(s^{\prime}_1,\ldots,s^{\prime}_m) \in \mathcal{S}$ such that $s^*_j \neq s^{\prime}_j.$
 There exists a repair scheme for the two simultaneous erasures $f(\bm{s^{\prime}})$ and $f(\bm{s}^*)$ with bandwidth at most
\[b = 2\left[ n-2 + (t-1)\left(\frac{n}{n_j}-2\right)\right].\]
\end{theorem}
\begin{proof}
Assume that the entries $f(\bm{s^{\prime}})$ and $f(\bm{s}^*)$ of the codeword $(f(\bm{s}_1),\ldots,f(\bm{s}_n))\in \mcc$ have been erased. By Remark~\ref{Trace Ker}, $\Delta_j:= \left\{\alpha \in K : \tr \left(\alpha (s^{\prime}_j - s^*_j  \right) = 0 \right\}$ has dimension $t-1$ as $\mathbb{F}_q$-vector space. Let $\{z_1, \dots, z_{t - 1}\}$ be an $\mathbb{F}_q$-basis for $\Delta_{j}$ and $z_t$ an element in $K$ such that $\{z_1, \dots, z_{t - 1}, z_t\}$ is an $\mathbb{F}_q$-basis for $K.$ Finally, let $\tau$ be an element of $\ker (\tau).$ 
We are ready to define the repair polynomials. Take
\[p_i(\bm{x}) = \tau\frac{\tr \left(z_i (x_j - s^*_j )\right)}{ \left(x_j - s^*_j \right)} \quad \text{ and } \quad
q_i(\bm{x}) = \frac{\tr\left(z_i (x_j - s^{\prime}_j )\right)}{ (x_j - s^{\prime}_j )}, \quad i \in [t]. \]
As $\mathcal{A} \cap L_j = \emptyset,$ the polynomials $p_i(\bm{x})$ and $q_i(\bm{x})$ define elements in the dual code $\mcc^\perp.$ Therefore, in a similar way to the proof of Theorem~\ref{21.06.17}, we obtain the following $2t$ equations:
\begin{align}
\lambda_{\bm{s}^*}p_{i}(\bm{s}^*)f(\bm{s}^*)+\lambda_{\bm{s^{\prime}}}p_{i}(\bm{s^{\prime}})f(\bm{s^{\prime}})&=\label{Eq.21.06.11_1}
-\sum_{\bm{s} \in \mathcal{S} \setminus\{\bm{s}^*, \bm{s^{\prime}}\}} \lambda_{\bm{s}} p_{i}(\bm{s})f(\bm{s}),\quad i\in[t],\\
\lambda_{\bm{s}^*}q_{i}(\bm{s}^*)f(\bm{s}^*)+\lambda_{\bm{s^{\prime}}}q_{i}(\bm{s^{\prime}})f(\bm{s^{\prime}})&=\label{Eq.21.06.11_2}
-\sum_{\bm{s} \in \mathcal{S} \setminus\{\bm{s}^*, \bm{s^{\prime}}\}} \lambda_{\bm{s}} q_{i}(\bm{s})f(\bm{s}),\quad i\in[t].
\end{align} 
By definition of the $p_i$'s and $q_i$'s, $p_i(\bm{s}^*)= \tau z_i$ and $q_i(\bm{s^{\prime}}) = z_i$ for $i \in [t].$ As $\{z_1, \dots, z_{t - 1}\}$ is an $\mathbb{F}_q$-basis for $\Delta_j$, $p_i(\bm{s^{\prime}}) = q_i(\bm{s}^*) = 0$ for $i \in [t-1],$ thus Equations~\ref{Eq.21.06.11_1} and \ref{Eq.21.06.11_2} become
\begin{eqnarray}
\lambda_{\bm{s}^*}\tau z_i f(\bm{s}^*)&=&\label{Eq.21.06.11_3}
-\sum_{\bm{s} \in \mathcal{S} \setminus\{\bm{s}^*, \bm{s^{\prime}}\}} \lambda_{\bm{s}} p_{i}(\bm{s})f(\bm{s}),\quad i\in[t-1],\\
\lambda_{\bm{s}^*}\tau z_t f(\bm{s}^*)+\lambda_{\bm{s^{\prime}}} \, p_{t}(\bm{s^{\prime}})f(\bm{s^{\prime}})&=&\label{Eq.21.06.11_4}
-\sum_{\bm{s} \in \mathcal{S} \setminus\{\bm{s}^*, \bm{s^{\prime}}\}} \lambda_{\bm{s}} p_{t}(\bm{s})f(\bm{s}),\\
\lambda_{\bm{s^{\prime}}} \, z_i f(\bm{s^{\prime}})&=&\label{Eq.21.06.11_5}
-\sum_{\bm{s} \in \mathcal{S} \setminus\{\bm{s}^*, \bm{s^{\prime}}\}} \lambda_{\bm{s}} q_{i}(\bm{s})f(\bm{s}),\quad i\in[t-1],\\
\lambda_{\bm{s}^*}q_{t}(\bm{s}^*) f(\bm{s}^*)+\lambda_{\bm{s^{\prime}}} \, z_t f(\bm{s^{\prime}})&=&\label{Eq.21.06.11_6}
-\sum_{\bm{s} \in \mathcal{S} \setminus\{\bm{s}^*, \bm{s^{\prime}}\}} \lambda_{\bm{s}} q_{t}(\bm{s})f(\bm{s}).
\end{eqnarray}
Observe that
\begin{eqnarray*}
\tr \left( \lambda_{\bm{s^{\prime}}} \, p_t( \bm{s^{\prime}} ) f( \bm{s^{\prime}} ) \right) & =  &
\tr \left( \lambda_{\bm{s^{\prime}}} \, \tau \frac{\tr \left(z_t (s^{\prime}_j - s^*_j )\right)}{(s^{\prime}_j - s^*_j )} f( \bm{s^{\prime}} ) \right)\\ 
& = &
 \tr \left(z_t (s^{\prime}_j - s^*_j )\right)
 \tr \left( \lambda_{\bm{s^{\prime}}} \, \frac{\tau}{(s^{\prime}_j - s^*_j ) } f( \bm{s^{\prime}} ) \right).
\end{eqnarray*}
As $\displaystyle \frac{\tau}{(s^{\prime}_j - s^*_j ) } \in \Delta_j,$ whose $\mathbb{F}_q$-basis is $\{z_1,\ldots, z_{t-1} \},$ there exist $\alpha_1, \ldots, \alpha_{t-1} \in \mathbb{F}_q$ such that previous equations imply that
\[\tr \left( \lambda_{\bm{s^{\prime}}} \, p_t( \bm{s^{\prime}} ) f( \bm{s^{\prime}} ) \right) =
 \tr \left(z_t (s^{\prime}_j - s^*_j )\right)
\sum_{i=1}^{t-1} \alpha _ i \tr \left( \lambda_{\bm{s^{\prime}}} \, z_i f( \bm{s^{\prime}} ) \right).\]
By Remark~\ref{Dual Bases Remark}, the element $f(\bm{s}^*)$ can be recovered from the $t$ traces $ \tr(\lambda_{\bm{s}^*}\tau z_i f(\bm{s}^*)).$ Thus, from last equation, and applying the trace function to both sides of Equations~\ref{Eq.21.06.11_3},~\ref{Eq.21.06.11_4} and \ref{Eq.21.06.11_5}, we get that the traces
$\tr(\lambda_{\bm{s}^*}\tau z_i f(\bm{s}^*)),$ for $i\in [t],$ can be obtained by downloading for every $\bm{s} \in \mathcal{S} \setminus\{\bm{s}^*, \bm{s^{\prime}}\},$ the elements 
$\tr (\lambda_{\bm{s}} p_{i}(\bm{s})f(\bm{s}))$ for $i\in [t],$ and $\tr (\lambda_{\bm{s}} q_{i}(\bm{s})f(\bm{s}))$ for $i\in[t-1].$
Finally, as $f(\bm{s}^*)$ has been already recovered, from Equation~\ref{Eq.21.06.11_6}, we can obtain $\tr (\lambda_{\bm{s^{\prime}}} \, z_t f(\bm{s^{\prime}})),$ and as a
consequence $f(\bm{s^{\prime}}),$ by downloading for every $\bm{s} \in \mathcal{S} \setminus\{\bm{s}^*, \bm{s^{\prime}}\},$ the elements $\tr (\lambda_{\bm{s}} q_{t}(\bm{s})f(\bm{s})).$

Therefore, both erasures $f(\bm{s^{\prime}})$ and $f(\bm{s}^*)$ can be recovered by downloading for every $\bm{s} \in \mathcal{S} \setminus\{\bm{s}^*, \bm{s^{\prime}}\},$ the elements 
$\tr (\lambda_{\bm{s}} p_{i}(\bm{s})f(\bm{s}))$ and $\tr (\lambda_{\bm{s}} q_{i}(\bm{s})f(\bm{s}))$ for $i\in[t].$ The bandwidth is a consequence of the proof of Theorem~\ref{21.06.17}, considering that now we need to download twice the information about $n-2$ elements, instead of only $n-1$ as in Theorem~\ref{21.06.17}.
\end{proof}

\begin{theorem}\label{21.06.18}
There exists a repair scheme for $\acar(\mathcal{S},\bm{k})$ that repairs two simultaneous erasures $f(\bm{s^{\prime}})$ and $f(\bm{s}^*)$ with bandwidth at most
\[b = 2\left[ \prod_{i = 1}^m n_i - 1 + (t-1) \left(\prod_{i = 1}^{m-1} n_i - 1\right) \right].\]
\end{theorem}
\begin{proof}
As $\bm{s}^* \neq \bm{s^{\prime}},$ there is $j\in [m]$ such that $s_j^* \neq s^{\prime}_j.$ The condition $k_i \leq n_i - q^{t - 1}$ on the definition of augmented Cartesian code 1 implies that $\mathcal{A}_{Car1}({\bm k}) \cap L_j = \emptyset,$ where $L_j= \{(a_1,\ldots,a_m) : n_j - q^{t - 1} \leq  a_j < n_j, a_i = n_{i-1}-1 \text{ for } i\neq j \}.$ Thus, the result follows from the proof of Theorem~\ref{21.06.16} and the fact that the length of the augmented Cartesian code 1, $n=\prod_{i = 1}^m n_i,$ is given by the cardinality of the Cartesian set $\mathcal{S}.$
\end{proof}

\begin{theorem}
There exists a repair scheme for $\acarr(\mathcal{S},\bm{k})$ that repairs two simultaneous erasures $f(\bm{s^{\prime}})$ and $f(\bm{s}^*)$ with bandwidth at most
\[b = 2\left[ \prod_{i = 1}^m n_i - 1 + (t-1) \left(\prod_{i = 1}^{m-1} n_i - 1\right) \right].\]
\end{theorem}
\begin{proof}
As $\bm{s}^* \neq \bm{s^{\prime}},$ there is $j\in [m]$ such that $s_j^* \neq s^{\prime}_j.$ By Remark~\ref{Trace Ker}, $\ker (\tr)$ has dimension $t-1$ as $\mathbb{F}_q$-vector space. Let $\{z_1, \dots, z_{t - 1}\}$ be an $\mathbb{F}_q$-basis for $\ker(\tr)$ and $z_t$ an element in $K$ such that $\{z_1, \dots, z_{t - 1}, z_t\}$ is an $\mathbb{F}_q$-basis for $K.$ Then we define the repair polynomials
\[p_i(\bm{x}) = z_1 \frac{\tr \left( \displaystyle \frac{z_i(x_j - s_j^*)}{s^{\prime}_j - s_j^*} \right)}{ \displaystyle \frac{x_j - s_j^*}{s^{\prime}_j - s_j^*}} \quad \text{ and } \quad
q_i(\bm{x}) = \frac{\tr \left( \displaystyle \frac{z_i(x_j - s^{\prime}_j)}{s_j^* - s^{\prime}_j} \right)}{\displaystyle  \frac{x_j - s^{\prime}_j}{s_j^* - s^{\prime}_j}}, \quad i \in [t]. \]
By definition of augmented Cartesian code 2, $k_i \leq n_i - q^{t - 1},$ for $i\in [m],$ thus the polynomials $p_i(\bm{x})$ and $q_i(\bm{x})$ define elements in the dual code $\acarr(\mathcal{S},\bm{k})^\perp.$
Observe that the polynomials $p_i$'s and $q_i$'s have the property that $p_i(\bm{s}^*)= z_1 z_i$ and $q_i(\bm{s^{\prime}}) = z_i$ for $i \in [t].$ By definition of the $z_i$'s, $p_i(\bm{s^{\prime}}) = z_1 \tr(z_i) = 0$ and $q_i(\bm{s}^*) = \tr(z_i) = 0,$ for $i \in [t-1].$ In addition, observe that
\begin{eqnarray*}
\tr \left( \lambda_{\bm{s^{\prime}}} \, p_t( \bm{s^{\prime}} ) f( \bm{s^{\prime}} ) \right) & =  &
\tr \left( \lambda_{\bm{s^{\prime}}} \, z_1 \tr (z_t) f( \bm{s^{\prime}} ) \right)\\ 
& = &
 \tr \left(z_t \right)
 \tr \left( \lambda_{\bm{s^{\prime}}} \, z_1 f( \bm{s^{\prime}} ) \right).
\end{eqnarray*}
Following the lines of the proof of Theorem~\ref{21.06.16}, we obtain that both erasures $f(\bm{s^{\prime}})$ and $f(\bm{s}^*)$ can be recovered by downloading for every $\bm{s} \in \mathcal{S} \setminus\{\bm{s}^*, \bm{s^{\prime}}\},$ the elements 
$\tr (\lambda_{\bm{s}} p_{i}(\bm{s})f(\bm{s}))$ and $\tr (\lambda_{\bm{s}} q_{i}(\bm{s})f(\bm{s}))$ for $i\in[t].$ Therefore, the result follows from the proof of Theorem~\ref{21.06.16} and the fact that the length of the augmented Cartesian code 2, $n=\prod_{i = 1}^m n_i,$ is given by the cardinality of the Cartesian set $\mathcal{S}.$
\end{proof}

\begin{remark}
Given certain circumstances, it is possible to extend the repair scheme described above for two erasures to three erasures and beyond.  This extension may be seen as analogous to the extension to three erasures in the Reed-Solomon case developed in \cite{Dau}. In particular, such a repair scheme for three erasures which all differ on the same coordinate $j$, $\bm{s^*}, \bm{s^{\prime}}, $ and $\bm{\tilde{s}}$ begins with finding the kernels of the following maps:
$$\tr(z(s'_j - s^*_j)), \tr(z(s'_j - \tilde{s_j})), \tr(z(s^*_j - \tilde{s}_j)).$$
Then, similar to the two erasure case, repair polynomials $\{p_1, \dots, p_t\}, \{q_1, \dots, q_t\}, $ and $\{r_1, \dots, r_t\}$ can be constructed which each evaluate to a basis element $\{z_1, \dots, z_t\}$ at the associated erased coordinates $\bm{s^*}, \bm{s^{\prime}}$, and $\bm{\tilde{s}}$ respectively. The basis chosen will be an extension of the basis for intersection of the three kernels.  This choice of basis combined with the properties of the trace function will guarantee that each repair polynomial will evaluate to $0$ at their non-associated erased coordinates, on all but two $i$. However, on these remaining $i$, the repair polynomials will evaluate to an element in the span of the outputs of other repair polynomials.  This will create a system of equations which can be solved given the output of two polynomials on these remaining $i$. For example, $p_{t - 1}(\bm{s^{\prime}})$ and $p_t(\bm{\tilde{s}})$ would be enough given the appropriate repair polynomial definitions.  Under particular circumstances, such as $t \mid \textrm{char}(K)$, these two outputs can be determined from the remaining nodes, and therefore produce $\tr(p_i(s)f(s))$, $\tr(q_i(s)f(s))$, and $\tr(r_i(s)f(s))$ at each erased coordinate for all $i's$.  Then, a typical linear exact repair scheme can proceed from there to fix all three erasures.

\end{remark}

\section{Comparisons and examples}\label{results}
The GW-scheme~\cite[Theorem 1]{GW} has the following parameters on the Reed-Solomon code $\rs(K, k)$:
length $n = |K|,$ dimension $\dim = k$ and bandwidth $b = n-1.$
Proposition~\ref{21.05.26} and Corollary~\ref{21.01.07} give the following parameters for the repair scheme on the augmented Reed-Muller code 1 ($\agrm$-scheme):
length $n = |K|^{m},$ $\dim = n {-} (\displaystyle \sqrt[m]{n} {-} k)^m$ and bandwidth
$b = |K|^m - 1 + (t-1)(|K|^{m-1}-1).$

It is clear that in general, the bandwidth of the $\agrm$-scheme may be much larger than the bandwidth of the GW-scheme, but the dimension and the length are also much larger. We now compare both schemes when the dimension and the base field $\mathbb{F}_q$ are the same. 

Assume $m$ divides $t$ and $t=m{t^*}.$ The GW-scheme and the $\agrm$-scheme repair the codes $\rs(\mathbb{F}_{q^t},k)$ and $\agrm(\mathbb{F}_{q^{t^*}}^m,k)$ when the dimensions are at most $q^t - q^{t-1}$ and $q^t - q^{t-m},$ respectively. An advantage of the $\agrm(\mathbb{F}_{q^{t^*}}^m,k)$ comes when a code with dimension $k^*$ between $q^t - q^{t-1}$ and $q^t - q^{t-m}$ is required. The restriction on the dimension of the GW-scheme implies that to employ an $\rs$ code, it must utilize an alphabet of size $q^{t+1}$ to achieve dimension $k^*.$ However, as the dimension of the code $\agrm(\mathbb{F}_{q^{t^*}}^m,k)$ can be up to $q^t - q^{t-m},$ there are values between $q^t - q^{t-1}$ and $q^t$ where we can still use $\agrm(\mathbb{F}_{q^{t^*}}^m,k)$, whose bandwidth can be lower. We show this in the following example.

\begin{example}\rm\label{21.07.02}
Assume that a code of dimension $k^*=648$ over a field of characteristic $3$ is required. Observe that $3^6 - 3^5 = 486 < k^* <  3^6 = 729.$ Over the field of size $3^6$, there is a Reed-Solomon code with dimension $648,$ but the GW-scheme is not applicable. Indeed, the requirement that the dimension is at most $n - q^{t-1} = 486$ is not satisfied. To resolve this, a larger field such as one of size $3^7=2187$ may be used. Given that the GW-scheme requires the dimension to be at most $n - q^{t-1}$, the RS code's length must then be bounded below by $648 + q^{t-1} = 1377$, meaning the bandwidth is at least 1376. The code $\agrm(\mathbb{F}_{3^{3}}^2,18)$ has dimension $k^*$ and according to Corollary~\ref{21.01.07}, bandwidth $b = |K|^m - 1 + (t-1)(|K|^{m-1}-1)=27^2-1+(2)(27-1)=780.$ As a consequence we obtain the following. Using $\rs$ codes and the GW-scheme, we obtain a code over $F_{2187}$, length $1377$, bandwidth $1376$ and dimension $648.$ Using $\agrm$ and the $\agrm$-scheme from Corollary~\ref{21.01.07}, we obtain a code over $F_{27}$, length $729$, bandwidth $780$ and dimension $648.$

Notice that applying the result in Corollary~\ref{21.01.07} to repair an erasure of $\agrm(\mathbb{F}_{3^{3}}^2,18)$ gives a bandwidth of $780$ whereas using~\cite[Theorem 2.5]{LMV}, the bandwidth is $837$~\cite[Example 4.1]{LMV}.
\end{example}

We can go further. As the following example shows, there are some values between $q^t - q^{t-1}$ and $q^t$ where an augmented Cartesian code can be used, but not an augmented Reed-Muller code.

\begin{example}\rm
Assume that a code of dimension $k^*=621$ over a field of characteristic $3$ is required. Observe that $3^6 - 3^5 = 486 < k^* <  3^6 = 729.$
As we explained on Example~\ref{21.07.02}, we can use a Reed-Solomon code and the GW-scheme, but the RS code's length must then be bounded below by $648 + q^{t-1} = 1377$, meaning the bandwidth is at least 1376.

Next, we consider whether we can use an augmented Reed-Muller code. According to Example~\ref{21.07.02}, the code $\agrm(\mathbb{F}_{3^{3}}^2,18)$ has dimension $648$ and bandwidth $780.$ If we increase $t$ or $m$, the bandwidth will increase. If we decrease $m$ from $2$ to $1,$ we are getting a RS code. So, the only option is to reduce $t.$ Over $\mathbb{F}_{3^2},$ in order to have a dimension $621$, we need $m=3.$ On this case, according to Proposition~\ref{21.05.26}, the dimension of $\agrm(\mathbb{F}_{3^{2}}^3,6)=721.$ By Corollary~\ref{21.01.07}, the bandwidth is $808.$

Now take $q=3, t=3, m=2, S_1= \mathbb{F}_{3^3}$ and $S_2= \mathbb{F}_{3^3}^*=S_1\setminus\left\{0\right\}.$ By Proposition~\ref{21.05.26}, the dimension of $\displaystyle \acar(S_1 \times S_2, (17,18)) = (26)(27) - (9)(9) = 621.$ Using Corollary~\ref{21.05.25}, we obtain that the bandwidth of $\displaystyle \acar(S_1 \times S_2, (17,18))$ is $(26)(27) - 1 (2)(26 - 1) = 621-1+2(25)=670.$

As a summary, if we want a code with dimension $621,$ using a $\rs$ code, we will have bandwidth 1376, using an $\arm$ code will we have bandwidth $780$, and using an $\acar$ we will have bandwidth $670.$
\end{example}

\begin{example}\rm
The $\agrm$-scheme may be compared with other repair schemes in the literature, such as the repair scheme for algebraic geometry codes~\cite{JLX}. By Corollary~\ref{21.01.07}, the augmented code $\agrm(\mathbb{F}_{2^{3}}^3,3)$ has length 512, dimension 448 and bitwidth $\log_{2} q (b) = |K|^m - 1 + (t-1)(|K|^{m-1}-1)= 8^3-1+(2)(8^2-1)=$ 637, whereas the Hermitian code of the same rate in ~\cite[Example 14]{JLX} requires a bitwidth of (3)(511)=1533 to repair an erasure. In addition, the $\agrm(\mathbb{F}_{2^{3}}^3,3)$ code is over $\mathbb{F}_8,$ while the Hermitian code is over $\mathbb{F}_{64}.$ An $\rs$ code of the same length and dimension requires a field of size at least 512 and a bitwidth of 1533.

Note that using Corollary~\ref{21.01.07} to repair an erasure of $\agrm(\mathbb{F}_{2^{3}}^3,3)$ gives a bandwidth of $637$ whereas using~\cite[Theorem 2.5]{LMV} provides the bitwidth is $847$~\cite[Example 4.2]{LMV}.
\end{example}

The $\arm$ codes will have greater repair bandwidth than the $\grm$ codes when $q$ increases. However, the expression of the bandwidth makes it difficult to immediately appreciate the improvement in rate gained by implementing the $\arm$ codes. Figure~\ref{fig:Comparison Plot} graphs the rate versus the repair bandwidth of the repair schemes of $\grm(\mathbb{F}_{5^4}^3, k)$, $\agrm(\mathbb{F}_{5^4}^3, k)$, and $\agrmm(\mathbb{F}_{5^4}^3, k),$ for all values of $k$ where the repair schemes developed in~\cite[Theorem 1]{GW} and Corollary~\ref{21.01.07} can be applied. The same figure demonstrates that $\grm$ codes admit repair schemes with much lower bandwidth than the $\arm$. However, it also reveals that the $\arm$ codes have significantly higher rates, increasing from at most $0.2$ to more than $0.99$.  Actual values can be found in Examples~\ref{5 example} and \ref{6 example}.

\begin{figure}[ht!]
\vskip -.3cm
\begin{center}
\includegraphics[width=8.5cm, height=4.5cm]{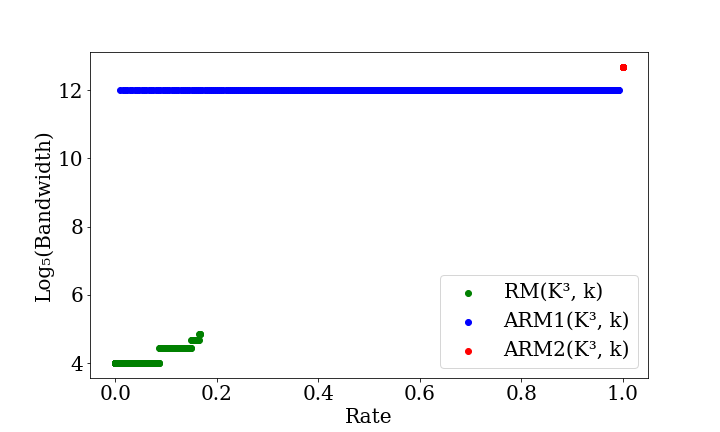}
\end{center}
\caption{Rate versus the repair bandwidth of the repair schemes of $\grm(\mathbb{F}_{5^4}^3, k)$, $\agrm(\mathbb{F}_{5^4}^3, k)$, and $\agrmm(\mathbb{F}_{5^4}^3, k),$ for all values of $k$ where the repair schemes developed in~\cite[Theorem 1]{GW} and Corollary~\ref{21.01.07} can be applied.}
\label{fig:Comparison Plot}
\vskip -.3cm
\end{figure}

\begin{example}\label{5 example}\rm
Let $q = 5$, $t = 4$, and $m = 3.$ The maximum $k$ for which $\grm(K^m, k)$ admits the repair scheme given in~\cite[Theorem III.1]{RM} is $623.$ The maximum $k$ for which $\agrm(K^m, k)$ and $\agrmm(K^m, k)$ admit the repair scheme given in Corollary~\ref{21.01.07} is $499.$ The code $\grm(K^m, 623)$ has rate $0.167$ and bandwidth $2496$. The code $\agrm(K^m, 499)$ has rate $0.992$ and bandwidth $245312496$. The code $\agrmm(K^m, 499)$ has rate $0.999998468$ and bandwidth $245312496$.
\end{example}
\begin{example}\label{6 example}\rm
The maximum $k$ for which $\grm(\mathbb{F}_{2^7}^5, k)$ admits the repair scheme given in~\cite[Theorem III.1]{RM} is $126.$ The maximum $k$ for which $\agrm(\mathbb{F}_{2^7}^5, k)$ and $\agrmm(\mathbb{F}_{2^7}^5, k)$ admit the repair scheme given in Corollary~\ref{21.01.07} is $63.$ The code $\grm(\mathbb{F}_{2^7}^5, 126)$ has rate $0.009002376$ and bandwidth $889$. The code $\agrm(\mathbb{F}_{2^7}^5, 62)$ has rate $0.96975$ and bandwidth $35970351097$. The code $\agrmm(\mathbb{F}_{2^7}^5, k)$ has rate $0.999999991$ and bandwidth $35970351097$.
\end{example}

Previous examples support the same conclusion. Reed-Muller codes admit repair schemes with superior bandwidth but have massively inferior rates when compared with the augmented codes.

We now compare augmented Reed-Muller and Cartesian codes when the length and the field $\mathbb{F}_{q^t}$ are both fixed.
\begin{example}\rm
Assume that an augmented code of length $n > 8$ over the field $\mathbb{F}_{2^3}$ is required. The augmented Reed-Muller code with minimum length greater than $8$ is the code $\agrm(\mathbb{F}_{2^3}^2, k),$ where $0 \leq k \leq 4.$ The bandwidth is $b = |K|^m - 1 + (t-1)(|K|^{m-1}-1)=8^2-1+(2)(8-1)=77.$ The augmented Cartesian code with minimum length greater than $8$ is the code $\acar(S_1\times S_2,(k_1,k_2)),$ where $n_1=|S_1|=4,$ $n_2=|S_2|=8,$ $k_1 = 0$ and $0 \leq k_2 \leq 4.$ The bandwidth is $b = \prod_{i = 1}^m n_i - 1 + (t-1) \left(\prod_{i = 1}^{m-1} n_i - 1\right)=(4)(8) - 1 + (2) (4 - 1)=37.$ 
\end{example}

Finally we study the case when Reed-Solomon, Reed-Muller and augmented Reed-Muller codes achieve their maximum rate.
\subsection{Maximum rates and asymptotic behavior}
Focusing on the improved rate, here we study the asymptotic behavior of the rate and the bandwidth rate $\displaystyle \frac{b}{nt},$ which represents the fraction of the codeword that is needed by the repair scheme to recover the erased symbol. We continue with the notation $K=\mathbb{F}_{q^t}.$

{\bf Reed-Solomon} The maximum $k$ for which $\rs(K, k)$ admits the repair scheme given in~\cite[Theorem 1]{GW} is $k^*= q^t - q^{t-1}.$ On this case, $\dim_K \rs(K, k) = q^t - q^{t-1}$ and the bandwidth at $k^*$ is $b^*=(q^t - 1).$ Thus
\begin{eqnarray*}
\lim _{t\to \infty}\frac{\displaystyle \dim_K \rs(K, k)}{n}
&=& \lim _{t\to \infty}\frac{q^t-q^{t-1}}{q^t-1} = \lim _{t\to \infty}\frac{q^{t-1} (q-1)}{q^{t-1}(q-\frac{1}{q^t})}=1-\frac{1}{q}.\\
\lim _{t\to \infty}\frac{\text{Bandwidth}}{tn}
&=& \lim _{t\to \infty}\frac{q^t-1}{t(q^t-1)} = \lim _{t\to \infty}\frac{1}{t}=0.
\end{eqnarray*}

{\bf Reed-Muller} The maximum $k$ for which $\grm(K^m, k)$ admits the repair scheme given in~\cite[Theorem III.1]{RM} is $k^*=q^t - 2.$ On this case, $\dim_K \grm(K^m, k^*)= \displaystyle {m + q^t - 2 \choose q^t - 2}$ and  bandwidth at $k^*$ is $b^*=(q^t - 1)t.$ Thus
\begin{eqnarray*}
\lim _{t\to \infty}\frac{\displaystyle \dim_K \grm(K^m, k^*)}{n} &=& 
\lim _{t\to \infty}\frac{\displaystyle {m + q^t - 2 \choose q^t - 2}}{\displaystyle q^{tm}}\\
&=& \lim _{t\to \infty}\frac{\left(m+q^t-2\right)!}{\left(q^t-2\right)! \,\, m! \,\, q^{tm}}\\
&=& \lim _{t\to \infty}\frac{\left(q^t-2+1\right)\cdots \left(q^t-2+m\right)}{m! \,\, q^{tm}}=\frac{1}{m!}\,.\\
\lim _{t\to \infty}\frac{\text{Bandwidth}}{tn}
&=& \lim _{t\to \infty}\frac{(q^t - 1)t}{tq^{tm}} = \lim _{t\to \infty}\frac{q^t - 1}{q^{tm}} = 0.
\end{eqnarray*}

{\bf Augmented Reed-Muller 1}
The maximum $k$ for which $\agrm(K^m, k)$ admits the repair scheme given in Corollary~\ref{21.01.07} is $k^* = q^t - q^{t-1}.$ On this case, $\dim_K \agrm(K^m, k^*)= q^{tm} {-} q^{(t - 1)m}$ and bandwidth at $k^*$ is $b^*= |K|^m - 1 + (t-1)(|K|^{m-1}-1).$ Thus
\begin{eqnarray*}
\lim _{t\to \infty}\frac{\dim_K \agrm(K^m, k^*)}{n} &=& 
\lim _{t\to \infty}\frac{q^{tm} {-} q^{(t - 1)m}}{q^{tm}} = 1-\frac{1}{q^{m}}\,.\\
\lim _{t\to \infty}\frac{\text{Bandwidth}}{nt}
&=& \lim _{t\to \infty}\frac{q^{tm} - 1 + (t-1)(q^{t(m-1)}-1)}{tq^{tm}}\\
&=& \lim _{t\to \infty} \left[\frac{q^{tm} - 1}{tq^{tm}} + \frac{(t-1)(q^{t(m-1)}-1)}{tq^{tm}}\right]=0.
\end{eqnarray*}

{\bf Augmented Reed-Muller 2}
The maximum $k$ for which $\agrmm(K^m, k)$ admits the repair scheme given in Corollary~\ref{21.01.07} is $k^* = q^t - q^{t-1}.$ On this case, $\dim_K \agrmm(K^m, k^*)= q^{tm} {-} m(q^{t-1} - 1) - 1$ and bandwidth at $k^*$ is $b^*= |K|^m - 1 + (t-1)(|K|^{m-1}-1).$ Thus
\begin{eqnarray*}
\lim _{t\to \infty}\frac{\dim_K \agrmm(K^m, k^*)}{n}
&=& \lim _{t\to \infty}\frac{q^{tm} {-} mq^{t-1} + m - 1}{q^{tm}}\\
&=& \lim _{t\to \infty}1 - \frac{m}{q^{t(m - 1) + 1}} + \frac{m - 1}{q^{tm}}=1.
\\
\lim _{t\to \infty}\frac{\text{Bandwidth}}{nt}
&=& \lim _{t\to \infty}\frac{q^{tm} - 1 + (t-1)(q^{t(m-1)}-1)}{tq^{tm}}\\
&=& \lim _{t\to \infty} \left[\frac{q^{tm} - 1}{tq^{tm}} + \frac{(t-1)(q^{t(m-1)}-1)}{tq^{tm}}\right]=0.
\end{eqnarray*}

{\bf Augmented Cartesian Codes}
The maximum $\bm{k}^*$ for which $\acar(\mathcal{S}, \bm{k}^*)$ admits the repair scheme given in Corollary~\ref{21.05.25} is $k^*_i = n_i - q^{t-1}.$ On this case, $\dim \acar(\mathcal{S}, {\bm k}^*) = \prod_{j = 1}^m n_j - \prod_{j = 1}^m (n_j - k_j)$ and bandwidth at ${\bm k}^*$ is $b^*= \prod_{i = 1}^m n_i - 1 + (t-1) \left(\prod_{i = 1}^{m-1} n_i - 1\right).$ Thus
\begin{align*}
    \lim_{t \to \infty} \frac{\textrm{Bandwidth}}{nt}
    &= \lim_{t \to \infty} \frac{\prod_{i = 1}^m n_i - 1 + (t - 1)(\prod_{i = 1}^{m - 1} n_i - 1)}{ t\prod_{i = 1}^{m} n_i}\\
    &= \lim_{t \to \infty} \frac{\prod_{i = 1}^{m}n_i - 1}{ t\prod_{i = 1}^m n_i} + \frac{t - 1}{t}\left(\frac{ \prod_{i = 1}^{m - 1}n_i}{\prod_{i = 1}^{m}n_i} - \frac{1}{\prod_{i = 1}^{m}n_i}\right)\\
    &= \lim_{t \to \infty}\frac{\prod_{i = 1}^{m}n_i - 1}{ t\prod_{i = 1}^mn_i} + \frac{t - 1}{t}\left(\frac{1}{n_m} - \frac{1}{\prod_{i = 1}^{m}n_i}\right).
\end{align*}
In the case where $n_m = \mathcal{O}(t)$, we have that this limit is $0$. 

Now we will discuss the limit of the rate of an Augmented Cartesian Code 1 as the extension degree $t$ approaches infinity through examples. We will find that varying the Cartesian evaluation set will result in augmented Cartesian codes with rate limits varying between $0$ and $1 - \frac{1}{q^m}$, even when taking the maximum allowable $k_j$'s.
\begin{example}\rm
Suppose we are in the case when the evaluation set $\mathcal{S} = S_1 \times \cdots \times S_m$ is such that $n_j = q^{t - 1} + 1$ for all $j \in [m].$ Consider the augmented Cartesian 1 code $\acar(\mathcal{S}, \bm{k}^*)$ with maximum rate. This happens when $\bm{k}^*=\bm{1}.$ The limit of the rate of this code as $t$ approaches infinity is
\begin{align*}
    \lim_{t \to \infty}\frac{\dim_K\acar(\mathcal{S}, \bm{1})}{n} &= \lim_{t \to \infty}\frac{\prod_{i = 1}^m n_i - \prod_{i = 1}^m(n_i - k_i)}{\prod_{i = 1}^m n_i}\\
    &= \lim_{t \to \infty}\frac{(q^{t - 1} + 1)^m - (q^{t - 1})^m}{(q^{t - 1} + 1)^m} = 0.
\end{align*}
\end{example}

\begin{example}\rm
Suppose we are in the case when the evaluation set $\mathcal{S} = S_1 \times \cdots \times S_m$ is such that $n_i=q^{t-1}$ for $i\in [m-1]$ and $n_m=2q^{t-1}.$ Consider the augmented Cartesian 1 code $\acar(\mathcal{S}, \bm{k}^*)$ with maximum rate. This happens when $k_i^*=0$ for $i\in [m-1]$ and $k^*_m=q^{t-1}.$ The limit of the rate of this code as $t$ approaches infinity is
\begin{align*}
    \lim_{t \to \infty}\frac{\dim_K\acar(\mathcal{S}, \bm{k})}{n} &= \lim_{t \to \infty}\frac{\prod_{i = 1}^m n_i - \prod_{i = 1}^m(n_i - k_i)}{\prod_{i = 1}^m n_i}\\
    &= \lim_{t \to \infty}1 - \frac{\prod_{i = 1}^m q^{t - 1}}{2\prod_{i = 1}^m q^{t - 1}} = 1 - \frac{1}{2} = \frac{1}{2}.
\end{align*}
\end{example}

\begin{example}\rm
Lastly, consider the case when $|S| = K^m$. As  this is an augmented Reed-Muller code, we obtain
\begin{align*}
\lim _{t\to \infty}\frac{\dim_K \acar(\mathcal{S}, \bm{k})}{n} = \lim _{t\to \infty}\frac{q^{tm} {-} q^{(t - 1)m}}{q^{tm}} = 1-\frac{1}{q^{m}}.
\end{align*}
\end{example}
A similar situation happens with the augmented Cartesian codes 2.
We summarize these findings in Table~\ref{tab:Table 2}.
\begin{table}[h]
    \begin{center}
    \scalebox{1.1}{
    \renewcommand{\arraystretch}{2.15}
    \begin{tabular}{ |c|c|c|c|c| } 
    \hline
    Code & Dimension & $\displaystyle \lim_{t\to \infty} \text{Rate}$& $\displaystyle \lim_{t\to \infty} \frac{\displaystyle b^*}{nt}$ \\ [0.2 cm]
    \hline
    $\rs(K,max)$ & $q^t - q^{t-1}$ & $\displaystyle 1-\frac{1}{q}$ & $0$ \\ [0.2 cm]
    \hline
    $\grm(K^m,max)$ & $\displaystyle {m + q^t - 2 \choose q^t - 2}$ & $\displaystyle \frac{1}{m!}$ & $0$ \\ [0.2 cm]
    \hline
    $\agrm(K^m, max)$ & $q^{tm} {-} q^{(t - 1)m}$ & $\displaystyle 1-\frac{1}{q^m}$ & $0$ \\ [0.2 cm]
    \hline
    $\agrmm(K^m, max)$ & $q^{tm} {-} m(q^{t-1}{-}1) {-}1$ & $1$ & $0$ \\ [0.2 cm]
    \hline
    $\acar({\bf \mathcal{S}}, max)$ & $\displaystyle\prod_{j = 1}^m n_j - \displaystyle\prod_{j = 1}^m (n_j - k_j)$ & \text{b/w} $0-1$ & $0$ \\ [0.2 cm]
    \hline
    $\acarr({\bf \mathcal{S}}, max)$ & $\displaystyle\prod_{i = 1}^m n_i - \displaystyle\sum_{i = 1}^m (n_i - k_i-1) -1$ & \text{b/w} $0-1$ & $0$ \\ [0.2 cm]
    \hline
    \end{tabular}}
    \end{center}
    \caption{Asymptotic behavior of the $\rs, \grm, \agrm$ and $\agrmm,$ when each achieves the maximum dimension so the associated repair scheme can be applied. The number $\displaystyle \frac{b^*}{nt}$ represents the fraction of the codeword that is needed by the repair scheme to recover an erased symbol.}
    \label{tab:Table 2}
\end{table}

 As expected, the augmented codes, which were designed to maximize the rate of the code, have a higher repair bandwidth as well, due to the trade-off between the rate of a code and the bandwidth of its associated repair scheme. In the end, neither of these schemes is objectively better than the other. Any potential user should opt to use the scheme that best deals with the parameter most important to their application, whether that be one that requires high rate codes or one that requires low bandwidth recovery.
\section{Conclusions}\label{conclusions}
In this paper, we introduce a new family of evaluation codes, called augmented Cartesian codes, along with repair schemes for single and certain multiple erasures. They can be designed to have higher rate than their traditional counterparts and include as a special case augmented Reed-Muller codes. In some circumstances, these repair schemes may have lower bandwidth and bitwidth than comparable algebraic geometry codes (such as Reed-Solomon or Hermitian codes). There are parameter ranges in which repairing Reed-Solomon codes may not be available, such as dimension between $q^t - q^{t-1}$ and $q^t$ over $\F_{q^t}$.  In some cases, augmented Reed-Muller codes may be designed along with repair schemes for single or pairs of erasures. More generally, we can use augmented Cartesian codes to provide high-rate codes with repair schemes for single erasures and certain pairs of erasures in those settings where the augmented Reed-Muller codes are not.

\end{document}